\documentclass[12pt]{article}

\makeatletter
\def\@seccntformat#1{\csname the#1\endcsname.\hspace*{0.5em}}
\makeatother

\usepackage{amsmath}
\usepackage{amsfonts}
\usepackage{amssymb}
\usepackage{graphicx,subfigure}
\usepackage{url}
\usepackage{xspace}
\usepackage{amsbsy}%
\usepackage{epsfig}%

\usepackage{theorem}%
\begingroup \makeatletter
\gdef\th@plain{\normalfont\itshape
  \def\@begintheorem##1##2{%
        \item[\hskip\labelsep \theorem@headerfont ##1\ ##2.]}%
\def\@opargbegintheorem##1##2##3{%
   \item[\hskip\labelsep \theorem@headerfont ##1\ ##2\ (##3).]}}
\endgroup

\newcommand{\EKhref}[2]{URL: \href{#1}{\nolinkurl{#2}}}
\usepackage[%
breaklinks%
,colorlinks%
,linkcolor=red
,anchorcolor=blue%
,pagecolor=blue%
,citecolor=blue%
,bookmarks=false%
]{hyperref}
\usepackage{breakurl}

\newcommand{\comment}[1]{}

\newcommand{\bb}{{\mathbf b}}
\newcommand{\I}{{\mathbf I}}
\newcommand{\III}{{\mathbb {I}}}
\newcommand{\ZZ}{{\mathbb Z}}

\newcommand{\RR}{{\mathbb R}}
\newcommand{\cc}{{\mathbf c}}
\newcommand{\ff}{{\mathbf f}}

\newcommand{\HH}{{\mathbf H}}
\newcommand{\xx}{{\mathbf x}}
\newcommand{\yy}{{\mathbf y}}

\newcommand{\qq}{{\mathbf q}}

\newcommand{\uu}{{\mathbf u}}
\newcommand{\vv}{{\mathbf v}}
\newcommand{\ww}{{\mathbf w}}
\newcommand{\mm}{{\mathbf m}}

\newcommand{\II}{{\mathcal I}}
\newcommand{\ldl}{PLDL${}^{\text{T}}$P${}^{\text{T}}$\xspace}

\newtheorem{theorem}{Theorem}
\newtheorem{lemma}{Lemma}

\newtheorem{definition}{Definition}
\newtheorem{problem}{Problem}
{\theorembodyfont{\upshape} \newtheorem{remark}{Remark}} 
{\theorembodyfont{\upshape} \newtheorem{example}{Example}} 

\makeatletter
\def\@yproof[#1]{\@proof{ #1}}
\def\@proof#1{\begin{trivlist}\item[]{\em Proof#1.}}
\newenvironment{proof}{\@ifnextchar[{\@yproof}{\@proof{}
}}{~$\Box$\end{trivlist}}
\makeatother

\title{%
Exact Safety Verification of Interval Hybrid Systems Based on Symbolic-Numeric Computation\raisebox{0ex}{*}%
}

\chardef\ttlb="7B 
\chardef\ttrb="7D 
\chardef\ttti="7E 
\author{%
Zhengfeng Yang$^a$,  Min Wu$^a$ and Wang Lin$^{b}$
\\
\\
\small\llap{$^a$}
Shanghai Key Laboratory of Trustworthy Computing\\
[-0.2ex] \small East China Normal University, Shanghai 200062, China
\\[-0.2ex]
\small\llap{$^b$}
College of Mathematics and Information Science \\
[-0.2ex] \small Wenzhou University, Zhejiang 325035, China
\\[-0.2ex]
\small {\ttfamily \{zfyang,mwu\}@sei.ecnu.edu.cn;
linwang@wzu.edu.cn}}
\begin{document}

\date{}

\maketitle
\def\thefootnote{\fnsymbol{footnote}}
\footnotetext[1]{%
\footnotesize
This material is supported in part by the National Natural Science
Foundation of China under Grants 91118007,61021004(Yang,Wu), and the
Fundamental Research Funds for the Central Universities under Grant
78210043(Yang,Wu). }

\begin{abstract}
In this paper, we address the problem of safety verification of
interval hybrid systems
in which the coefficients are intervals instead of explicit numbers.
A hybrid symbolic-numeric method, based on SOS relaxation and
interval arithmetic certification, is proposed to generate exact
inequality invariants for safety verification of interval hybrid
systems. As an application,
an approach is provided to verify
safety properties of
non-polynomial hybrid systems.
Experiments on the benchmark hybrid systems
are given to illustrate the
efficiency of our
method.
\end{abstract}

\section{Introduction}

As a tool of modelling cyber-physical systems,
hybrid systems are dynamical systems governed by interacting discrete and
continuous dynamics. The continuous dynamics of a hybrid system is specified by
differential equations, and for discrete transitions, the hybrid
system changes state instantaneously and possibly discontinuously.
Among the most important research issues in formal analysis of
hybrid systems are
{\em safety},
i.e., deciding whether a given property 
holds in all the reachable states, and  its dual problem {\em
reachability}, i.e., deciding if there exists a trajectory starting
from the initial set that reaches a state satisfying the given
property. Due to the infinite number of possible states in state
spaces, safety verification and reachability analysis of hybrid
systems presents a challenge.
For general (exact) 
 hybrid systems, some well-established techniques
\cite{CarbonellTiwari2005, GulwaniTiwari2008, liu2011complete,
 Platzer:2009, SSM:2008, prajna2007safety, tiwari2003approximate,
sturm2011verification} based on invariant generation have been
proposed for safety verification of the systems. However,
when applying these techniques,
one can not avoid numerical errors or may suffer from high complexity.
To take advantage of the efficiency of numerical computation and the
error-free property of symbolic computation, we proposed
in~\cite{WY11} a hybrid symbolic-numeric method via
exact sums-of-squares (SOS) representation
to construct differential
invariants for continuous dynamic systems,
and  generalized in~\cite{LWYZ11,yang2012exact} the idea for safety
verification of polynomial hybrid systems.

A common assumption made on
hybrid systems
is that the coefficients of the involved equations are specific values.
In practice, however, due to the increasing
complexity of modern systems, some disturbance and modeling errors
may be contained in the system description, and, in addition, there may be
noisy and inexact data involved in the realistic problem.
All these factors may contribute to inexactness of the data used to describe the hybrid systems.
To take
this uncertainty
into account, it would be more reasonable and appropriate to use intervals
rather than
concrete but inexact data to represent the
hybrid systems.
This motivates us to introduce
the notion of {\em interval polynomial hybrid systems}, by which we
mean the
differential equations in 
hybrid systems are represented as
polynomials with interval coefficients.

In this paper, we consider safety verification of interval
polynomial hybrid systems, i.e., deciding whether none of
trajectories of an interval hybrid system starting from the initial
set can enter some unsafe regions in the state spaces. In
\cite{yang2012exact} we applied a symbolic-numeric computation
method, based on bilinear matrix inequality (BMI) solving and exact
SOS polynomials representations, to
 deal with
 {\it exact} safety verification for polynomial hybrid systems.
  In this paper, we extend the techniques in \cite{yang2012exact} to generate exact invariants for
  verifying interval hybrid systems.
  The idea lies in applying interval arithmetic
  to verify positive semidefiniteness
  of interval matrices and
 existence of solutions to interval polynomial equations.
As an application,
we apply the above approach to
verify {\em safety of non-polynomial hybrid systems}
by relaxing continuous dynamics of non-polynomial forms to those of
interval polynomial forms,
and then studying safety of the latter system whose set of
trajectories contains that of the original non-polynomial system.

The contributions of our paper are as follows.
First, an approach is proposed to verify safety
property of an interval hybrid system, therefore, safety
property is guaranteed for an arbitrary hybrid system within the
given interval system.
Moreover, our approach can generate exact invariants instead of approximate ones, overcoming
the unsoundness of verification caused by
numerical errors \cite{PC:07}. And in comparison with some symbolic
approaches based on qualifier elimination, our approach is
more efficient and practical, because parametric polynomial
optimization problem based on SOS relaxation can be solved in
polynomial time theoretically.
 Second, a key problem we consider in safety verification is that of determining nonnegativity of
interval multivariate polynomials,
which is a
fundamental problem in real algebraic geometry.
Thirdly, for a non-polynomial
function,
we propose a rigorous polynomial approximation method to compute its
approximate polynomial with polynomial lower and upper bounds of the
interpolation error. Compared with the classical Taylor
approximation, the polynomial bounds we give is much sharper.

The rest of the paper is organized as follows.
In Section \ref{sect:notion}, we introduce
some notions related to interval hybrid systems. Section
\ref{sec:Nonnegativity} is devoted to determining nonnegativity
of interval multivariate polynomials.
In Section \ref{sect:Interval}, two techniques which combine SOS
relaxation with interval arithmetic are proposed to generate
invariants
of interval hybrid systems with
small and large radii, respectively. As an application, safety
verification of 
non-polynomial hybrid systems is discussed
in Section \ref{nonpolynomial}.
Section~\ref{sect:conclusion} concludes the paper.


\section{Interval Hybrid Systems and Safety Verification}\label{sect:notion}

Let us first review some notions of general hybrid
systems~\cite{Henzinger1996,SSM:2008}.

\begin{definition}[Hybrid System]\label{def:hybrid_system}
 A {\em hybrid system} is a tuple~$\HH: \langle V, L,
  {\mathcal T},  { \Theta}, \mathcal D, \Psi,  \ell_0\rangle$
  with
\begin{itemize}
\item $V=\{x_1,...,x_n\}$, a set of real-valued system {\em variables};

\item $L$, a finite set of locations;

\item $\ell_0\in L$, the {\em initial location};

\item $\mathcal T$, a set of
transitions. Each transition
$\tau:\langle \ell,\ell',g_{\tau}, \rho_{\tau} \rangle \in \mathcal
T$ consists of a prelocation~$\ell\in L$, a postlocation~$\ell'\in
L$, the guard condition $g_{\tau}$ over  $V$, and an
assertion~$\rho_\tau$ over~$V\cup V^\prime$ representing the
next-state relation, where $V^\prime=\{x_1',...,x_n'\}$ denotes the
next-state variables;

\item $\Theta$, an assertion specifying the {\em initial condition};

\item $\mathcal D$, a map that associates each location~$\ell\in L$ to a {\em differential rule}
(a.k.a.\ a {\em vector field})~$\mathcal D(\ell)$,
an autonomous system 
$\dot x_i= f_{\ell,i}(V)$ for each~$x_i\in V$,  written briefly as
$\dot \xx= \ff_{\ell}(\xx)$;

\item $\Psi$, a map that
sends
$\ell\in L$ to a {\em
location invariant}~$\Psi (\ell)$, an assertion
over~$V$.

\end{itemize}

\end{definition}


In reality, due to measuring errors or disturbance, the data involved in the systems may be inaccurate.
It is then reasonable to consider hybrid systems in which some data are
given as interval estimates rather than specific values, the so-called
{\em interval hybrid systems}.
Similar to
Definition~\ref{def:hybrid_system}, an interval hybrid
system~$\I\HH$
is defined to be a tuple
$$ \langle V, L, {\mathcal T}, {
\Theta}, [\mathcal D], \Psi, \ell_0\rangle,
$$
where
$V$,
$L$,
$\mathcal T$,
$\Theta$,
$\Psi$, $\ell_0$ are the same as in
Definition~\ref{def:hybrid_system}, while $[\mathcal D]$ represents
a map sending each location $\ell\in L$ to an {\em interval}
differential rule $[\mathcal D(\ell)]$
of the form
$$
\dot x_i = [f]_{\ell,i}(\xx)\quad i=1,\dots, n,
$$ by which we mean $[f]_{\ell,i}(\xx)$
is a real function with interval coefficients;
for brevity, we write $[\mathcal D(\ell)]$ as $\dot \xx = [\ff]_{\ell}(\xx)$;
%
For more details on interval arithmetic, please refer to Appendix A.


A hybrid system $\HH:\langle V, L, {\mathcal T}, { \Theta}, \mathcal
D, \Psi, \ell_0\rangle$ is said to be {\em within} an interval hybrid system $\I\HH: \langle V, L,
{\mathcal T}, {\Theta}, [\mathcal D], \Psi, \ell_0\rangle$ if
$f_{\ell, i}(\xx)\in [f]_{\ell,i}(\xx)$ for each~$\ell\in L$ and
$i=1,\dots,n$, or written briefly as $D(\ell)\in [D](\ell)$.

In this paper, we will
mainly study safety
verification of interval hybrid systems.
Recall that a hybrid system is said to be {\em safe} if none of the trajectories starting from any state in the initial set can evolve to an unsafe region.
Similarly, given  a prespecified unsafe region $X_u\subset
\RR^n$,
an interval
system $\I\HH: \langle V, L, {\mathcal T}, { \Theta}, [\mathcal D],
\Psi, \ell_0\rangle$ is said to be {\em safe} if every hybrid system within $\I\HH$ is safe.
This is to say, none of the trajectories of interval hybrid system $\I\HH$
starting from any state in the initial set can evolve to $X_u$, or,
equivalently, any state in $X_u$ is not reachable.



Recall that 
an invariant of a hybrid system~$\HH$
is an
over-approximation of all the reachable states of the system $\HH$. Since generating invariants of arbitrary form for
hybrid systems is computationally
hard,
the usual technique
is to compute inductive invariants. It is shown in \cite{yang2012exact} that safety verification
of general
hybrid systems
can be reduced to finding inductive invariants
(a.k.a.\ {\em barrier certificates} in \cite{prajna2007safety})
 of
hybrid systems,
as
described in the following theorem.
\begin{theorem}\label{thm:location_inv}
[
\cite{prajna2007safety}, \cite{WY11}] 
Let $\HH: \langle V, $ $L$, ${\mathcal T},$  ${\Theta}, \mathcal D,
\Psi,  \ell_0\rangle$ be a general hybrid system. Suppose
that for each location $\ell\in L$, there exists a function
$\varphi_{\ell}(\xx)$ such that 
\begin{description}
\item[(i)] $\Theta \models \varphi_{\ell_0}(\xx) \geq 0;$
\item[(ii)] $\varphi_{\ell}(\xx)\geq0 \wedge  
g(\ell, \ell') \wedge \rho(\ell,\ell') \models
\varphi_{\ell'}(\xx')\geq 0$,\,
    for any transition $\langle \ell, \ell', g, \rho\rangle$ going
out of~$\ell$;
\item[(iii)] $\varphi_{\ell}(\xx)\geq0  
\wedge\Psi(\ell)\models \dot{\varphi_{\ell}}(\xx)> 0,$ where $\dot
\varphi_\ell(\xx)$ denotes the {\em Lie-derivative} of
$\varphi_\ell(\xx)$ along the vector field $\mathcal D(\ell)$, i.e.,
$\dot{\varphi_{\ell}}(\xx)=\sum_{i=1}^n\frac{\partial\varphi_{\ell}}{\partial
x_i}  f_{\ell,i}(\xx)$.
    \end{description}
Then $\varphi_{\ell}(\xx)\geq0$ is an (inductive) invariant of the hybrid system
$\bf H$ at location~$\ell$. If, moreover,
\begin{description}
\item[(iv)] $X_u(\ell)\models \varphi_{\ell}(\xx) < 0 \quad \mbox{for any } \ell \in L,$
\end{description}
then the safety of the system $\bf H$ is guaranteed.
\end{theorem}

%
%
The notion of inductive invariants can be generalized for interval hybrid systems, as
defined in the following
\begin{definition}
[Inductive invariant]\label{def:inductiveinv_int}
For an interval hybrid system $\I\HH :\langle V, {L}, {\mathcal T},
{\Theta}, [{\mathcal D}], \Psi, \ell_0\rangle$, an {\em inductive
assertion map} $\II$ of $\I\HH$  is a map that associates with each
location~$\ell\in L$ an assertion~$\II(\ell)$ that holds initially
and is preserved by all discrete transitions and continuous flows
of~$\I\HH$. More formally, the map $\II$ satisfies the following
requirements:
\begin{description}
\item {\bf[Initial]} $\Theta
\models \II(\ell_0).$
\item {\bf[Discrete Consecution]} For each discrete transition $\tau:\langle
\ell,\ell',g_{\tau}, \rho_{\tau}\rangle$ starting from a state
satisfying~$\II(\ell)$, taking~$\tau$ leads to a state
satisfying~$\II(\ell')$, i.e., $\II(\ell) \wedge g_\tau \wedge
\rho_\tau \models\II(\ell')$ where~$\II(\ell')$ represents the
assertion~$\II(\ell)$ with the current state variables~$x_i$'s replaced by the next state variables~$x_i'$'s,
respectively.

\item {\bf[Continuous Consecution]} For location~$\ell \in L$ and states $\langle \ell, \xx_1\rangle$, $\langle \ell, \xx_2\rangle$ such that~$\xx_2$ evolves
from~$\xx_1$ according to
any differential rule~$\mathcal{D}(\ell)\in [\mathcal{D}](\ell)$, if~$\xx_1\models \II(\ell)$ then~$\xx_2\models
\II(\ell)$.
\end{description}
\end{definition}

The difference between inductive invariants of interval hybrid
systems and those of general hybrid systems lies in that for continuous consecution, any differential rule contained in the
interval differential rule must be considered.
Then
Theorem \ref{thm:location_inv} can be modified for
verifying safety of interval hybrid systems, as described in the following.
\begin{theorem}\label{thm:inv_int}
Let $\I\HH: \langle V, $ $L$, ${\mathcal T},$  ${\Theta}, [{\mathcal
D}], \Psi,  \ell_0\rangle$ be an interval hybrid system. Suppose
that for each $\ell\in L$, there exists a function
$\varphi_{\ell}(\xx)$ satisfying the conditions (i-ii) in Theorem
\ref{thm:location_inv}, and
\begin{description}
\item[(iii')] $\varphi_{\ell}(\xx)\geq0  
\wedge\Psi(\ell)\models \dot{\varphi_{\ell}}(\xx)> 0,$ here $\dot
\varphi_\ell(\xx)$ denotes the {\em Lie-derivative} of
$\varphi_\ell(\xx)$ along
any differential rule $\mathcal
D(\ell)
\in [\mathcal D(\ell)]$,
i.e.,
$\dot{\varphi_{\ell}}(\xx)=\sum_{i=1}^n\frac{\partial\varphi_{\ell}}{\partial
x_i}  f_{\ell,i}(\xx)$,
for any $
f_{\ell,i}(\xx)\in[f_{\ell,i}](\xx)$.
\end{description}
Then $\varphi_{\ell}(\xx)\geq0$ is an (inductive) invariant of the
interval hybrid system $\I\HH$ at location~$\ell$. If, moreover, the
condition (iv) in Theorem \ref{thm:location_inv} is satisfied,
then the safety of the
system $\I\HH$ is guaranteed.
\end{theorem}

In
our preceding papers \cite{LWYZ11, yang2012exact},
a symbolic-numeric method based on SOS relaxation, Gauss-Newton
refinement and rational vector recovery techniques
is proposed
to generate
polynomial
inequality
invariants
$\varphi_{\ell}(\xx)\ge 0$ at each
location $\ell\in L$ for general polynomial hybrid systems. This method can not be applied directly
on
interval hybrid systems.
In the sequel, we will
combine
BMI solving
with interval arithmetic to compute  polynomial invariants
$\varphi_{\ell}(\xx)\geq 0$
which satisfy
conditions in Theorem
\ref{thm:inv_int}.
For brevity, we will abuse the
notation~$\varphi_{\ell}(\xx)$ to represent both the
polynomial~$\varphi_{\ell}(\xx)$ and the
invariant~$\varphi_{\ell}(\xx)\geq 0$.


\section{Nonnegativity of Interval Polynomials}\label{sec:Nonnegativity}
To determine whether a polynomial inequality
$\varphi_{\ell}(\xx)\geq 0$ is an invariant of an interval hybrid
system, by Theorem~\ref{thm:inv_int} (iii') it suffices to
decide whether a multivariate polynomial $\dot\varphi_\ell(\xx)$
with interval coefficients is
positive semidefinite. In the sequel, we will call
a polynomial with interval coefficients an interval polynomial.
Denote by $\III\RR[\xx]$ the set of interval multivariate
polynomials in $\xx$. The first problem to be investigated is
the following
\begin{problem}\label{prob:1}
Given an interval polynomial $[\psi](\xx)\in\III\RR[\xx]$,
verify whether it is positive semidefinite, or the validity
of the interval inequality
\begin{equation*}
  [\psi](\xx)\geq 0,\,\,\forall \xx\in\RR^n.
\end{equation*}
\end{problem}
It is well known that the
problem of testing
positive semidefiniteness of  real polynomials is
NP-hard (when the degree is at least four). As stated in Appendix B,
a sufficient condition for a multivariate polynomial to be
positive semidefinite
is that there exists an SOS
polynomial (or rational function) representation.
In \cite{KLYZ:08, KYZ09,
peyrl2008computing},
some symbolic-numeric
methods were proposed to determine
whether a multivariate polynomial $\psi(\xx)$ with rational coefficients is positive semidefinite by computing its
exact SOS representations,
or equivalently, to determine if there exists a symmetric matrix
${W}\in \RR^{k\times k}$ satisfying
exactly
\begin{eqnarray}\label{interval:SOS}
             \psi(\xx)=\mm(\xx)^{T}\cdot {W} \cdot \mm(\xx) \text{ and } {W}\succeq 0,
\end{eqnarray}
where $W\succeq 0$ denotes that $W$ is positive semidefinite.
These methods cannot be applied directly to verifying
positive semidefiniteness of an interval polynomial $[\psi](\xx)
\in\III\RR[\xx]$, since there are
infinitely many
polynomials
in the
interval,
and it is impossible to provide certificates of SOS representations for infinitely many polynomials in $[\psi](\xx)$.
For {\bf Problem~\ref{prob:1}}, we will only prove existence of SOS representations for
polynomials in $[\psi](\xx)$.
%
This problem can be further distinguished into two cases
according to the radii of the coefficient intervals:
the coefficient intervals of $[\psi](\xx)\geq 0$ are all smaller
(resp.\ larger) than the given threshold.
%
In the sequel, we will describe how to deal with
the former case,
and
the latter case will be discussed in
subsection~\ref{subsec:large}.

Let $[{W}]$ be an interval matrix such that $[{W}]\succeq 0$, i.e.,
every matrix within $[{W}]$ is positive semidefinite. If for any
polynomial $\psi(\xx)$ within $[\psi](\xx)$, there exists a matrix
${W}\in[{W}]$ such that the condition (\ref{interval:SOS}) holds
exactly, then we have $[\psi](\xx)\geq 0$. Thus the first case of {\bf Problem~1}
can be transformed into the problem of finding an interval matrix
$[{W}] \succeq 0$ for
 $[\psi](\xx)$.

Suppose that there exists an approximate SOS decomposition of the
mid-point function $\text{mid}\psi(\xx)$ $\in[\psi](\xx)$:
\begin{equation}\label{SOS:approx}
\text{mid}\psi(\xx)\approx \mm(\xx)^{T}\cdot \widehat{W} \cdot
\mm(\xx)
\end{equation}
where $\widehat W\succeq 0 $.
Having $\widehat{W}$, we will consider how to compute
an interval
matrix $[{W}]\succeq0$ of minimal radius, 
such that $[W]$ contains $\widehat{W}$ and
for any $\psi(\xx)\in[\psi](\xx)$
there always exists a matrix ${W}\in[{W}]$ satisfying the
condition (\ref{interval:SOS})  exactly.
Considering whether the matrix $\widehat W$ is of full rank, there are two cases to be
addressed.

\subsection{$\widehat{W}$ is of full rank}\label{subsec:fullrank}


Suppose that $\widehat{W}$
in~(\ref{SOS:approx})
 is of full rank {\it numerically},
namely, the minimal eigenvalue of $\widehat{W}$ is greater than the
given tolerance $\tau>0$. Let $$[{W}]:=\widehat{W}+[\Delta W] $$ be
an interval matrix
perturbed  from~$\widehat W$ where $[\Delta
W]\in\III\RR^{k\times k}$.
If, for
any $\psi(\xx)\in[\psi](\xx)$,
there exists a matrix $\Delta W\in[\Delta W]$ which satisfies
\begin{eqnarray}\label{interval:SOS2}
\left.\begin{array}{l@{}l}
             & \psi(\xx)=\mm(\xx)^{T}\cdot (\widehat{W}+\Delta W) \cdot \mm(\xx) ,
\end{array}\right.
\end{eqnarray}
and $\widehat{W}+\Delta W \succeq 0$ exactly, then we have
$[\psi](\xx)\geq 0$.
Since
$\widehat{W}$ is positive definite and of full rank, according to
matrix perturbation theory
we have $\widehat{W}+[\Delta W] \succeq 0$ as long as the radius of
interval matrix $[\Delta W]$ is small enough.

We first consider how to construct an interval matrix $[\Delta W]$
with small radius, which satisfies the condition
(\ref{interval:SOS2}). Comparing the coefficients of terms
on both sides of (\ref{interval:SOS2}) gives rise to the following
underdetermined linear system with the entries of $\Delta W$ as
unknowns $\ww$:
\begin{equation*}\label{system_underdetermined}
A\cdot\ww=[\vv],
\end{equation*}
where
$A\in\RR^{s\times r}$ with $s\in\ZZ^+$ and $r=\frac{k(k+1)}{2}$,
$\ww\in\III\RR^r$ is a vector composed of
columnwise entries of the symmetric matrix $\Delta W$, and
$[\vv]\in\III\RR^s$ is
the coefficient vector of the interval polynomial
$[\psi](\xx)-\mm(\xx)^{T}\cdot \widehat{W} \cdot \mm(\xx)$. Our goal
is to compute a minimal $2-$norm
interval vector $\ww$ satisfying
$A\cdot \ww=[\vv]$.
The above problem is then transformed into the following interval least squares problem:
$$
  \Sigma=\min\{\|\ww\|_2:\,A\cdot\ww=\vv \text{ for some } \vv\in[\vv]\}.
$$
Using the method \cite{rump2010verification} for solving interval
linear systems, we can obtain a solution
$[\ww']\in\III\RR^r$ of $\Sigma$ and
therefore the associated solution $[\Delta W]$ of (\ref{interval:SOS2}) of minimal radius.
Then the
remaining task is to verify
whether the interval matrix
$\widehat{W}+[\Delta W]$ is positive semidefinite.
The following theorem provides
such a computational criterion.
\begin{theorem}\label{thm:positive}
\cite[Theorem 4]{rohn1994positive}
Let $[W]$ be a
symmetric interval matrix and $[W]=[\widehat{W}-\Delta
W, \widehat{W}+\Delta W]$ be its midpoint-radius form.
Suppose that $\rho(\Delta W)$ is the spectral radius of $\Delta W$ and
$\lambda_{min}(\widehat{W})$ is the minimum eigenvalue of
$\widehat{W}$. If $
   \rho(\Delta W)\leq \lambda_{min}(\widehat{W}),$
then $[{W}]$ is positive semidefinite. Moreover, if
$\rho(\Delta W)< \lambda_{min}(\widehat{W})$ then $[{W}]$ is
positive definite.
\end{theorem}

We give an example  to illustrate the above method.

\begin{example}
Verify $[\psi](\xx)\geq 0$ where
\begin{multline*}\small
[\psi](\xx)= 0.9574- 1.9362{x_1}- 0.3404{x_2}+[ 1.1852,
 1.2593]x_1^{2}\\-[0.4237,
0.4576]{x_1}\,{x_2}+[ 1.125, 1.2083]x_2^{2}.
\end{multline*}
For the mid-point function $\text{mid}\psi(\xx)$, we compute its
approximate Gram matrix representation
$\text{mid}\psi(\xx)\approx\mm(\xx)^T\cdot \widehat{W}\cdot
\mm(\xx)$ where
\begin{equation*}
 \mm(\xx)=\left( \begin {array}{ccc}  1\\
\noalign{\medskip}x_1\\ \noalign{\medskip}x_2\end {array} \right),
\widehat{W}=\left( \begin {array}{ccc}  0.9574&- 0.9681&- 0.1702\\
\noalign{\medskip}- 0.9681& 1.2222&- 0.2203\\ \noalign{\medskip}-
 0.1702&- 0.2203& 1.1667\end {array} \right).
\end{equation*}
 It is easy to check that $\widehat{W}$ is of full rank. By solving an
associated
interval linear system, we obtain the symmetric interval matrix
$[W]$ as follows:
\begin{equation*}
\left( \begin {array}{ccc} {[0.9574,0.9575] }&{[-0.9681,-0.9680]}&{[-0.1703,-0.1702]}\\
\noalign{\medskip}{[-0.9681,-0.9680]}&{[
1.1851,1.2593]}&{[-0.2289,-0.2118]}
\\ \noalign{\medskip}{[-0.1703,-0.1702]}&{[-0.2289,-0.2118]}&{[1.1388, 1.1945]}\end {array}
 \right).
\end{equation*}
For the midpoint-radius form of $[W]$, we obtain $0.0422=\rho(\Delta
W)< \lambda_{min}(\widehat{W})=0.0461.$ According to
Theorem~\ref{thm:positive}, $[W]$ is positive definitive,
which proves
$[\psi](\xx)\geq0$.
$\hfill \Box$
\end{example}

\subsection{$\widehat{W}$ is singular}\label{subsec:singular}

When the matrix $\widehat{W}$ is singular
or near to a singular matrix, the perturbed matrix of
$\widehat{W}$ may not be positive semidefinite. Therefore,
the
method in subsection~\ref{subsec:fullrank} does not apply to the case where $\widehat{W}$ is numerically singular.

By expanding the quadratic representation, the equation
(\ref{SOS:approx}) can be rewritten as
\begin{equation*}
\text{mid}\psi(\xx)\approx \displaystyle{\sum_{i=1}^{l}
\bigg(\sum_{\alpha}} \hat{q}_{i,\alpha} \xx^{\alpha}\bigg)^2,
\end{equation*}
where $l$ is the rank of $\widehat{W}$.
Next we will verify, for each
$\psi(\xx)\in[\psi](\xx)$,
there exist $q_{i,\alpha}\in\RR$
such that
\begin{equation}\label{constant_sos}
\psi(\xx)= \displaystyle{\sum_{i=1}^{k} \bigg(\sum_{\alpha}}
q_{i,\alpha} \xx^{\alpha}\bigg)^2
\end{equation}
holds exactly. Let $\qq$ be a vector composed of all the
$q_{i,\alpha}$.
Comparing the terms of both sides of (\ref{constant_sos})
gives rise to a nonlinear system
of the form
\begin{equation}\label{int:underdetermined1}
   F(\qq)-[\vv]=0,
\end{equation}
where $F:\RR^r\rightarrow\RR^s$ with $r$ the size of $\qq$, and
$[\vv]\in\III\RR^s$ is an interval vector consisting of coefficients
in $[\psi](\xx)$. Note that $F(\mathbf{0})=\mathbf 0$. Hence, the problem
of determining $[\psi](\xx)\geq 0$
is
equivalent to that of verifying existence of real roots
of the underdetermined interval nonlinear system~(\ref{int:underdetermined1}).
The latter problem can be solved
in  two ways:
one is
based on
existence of real roots for particular interval square nonlinear systems,
and the other for particular interval underdetermined  nonlinear systems.
The details of these two methods are given in  Appendix C.

\begin{remark}
If we find  a verified  real solution to system (\ref{int:underdetermined1}),
 then
$\psi(\xx)\geq 0$ for each $\psi(\xx) \in [\psi](\xx)$. However, the
opposite is not true, i.e., even if $[\psi](\xx)\geq 0$ it is not
guaranteed that the above methods can prove existence of real roots
of (\ref{int:underdetermined1}).
\end{remark}

\section{Safety Verification of Interval Hybrid Systems}\label{sect:Interval}

In this section, we study how to verify safe properties of
an interval hybrid system.
Two techniques will be used depending on the radii of the occurred
intervals in the given interval
hybrid system. If the radii of the intervals are all larger than a given threshold, we transform the
interval hybrid system into an uncertain hybrid system by
replacing the intervals with some uncertainties and then generalize the method
in~\cite{LWYZ11,yang2012exact}, which is based on SOS relaxation and rational
vector recovery, to
compute exact invariants of the uncertain hybrid system. If the
radii of the involved intervals are all less than the given
threshold, we will apply  the interval verification approach in
Section~\ref{sec:Nonnegativity}. For the more general case, when the
interval hybrid system contains both intervals of radii smaller than
and those of radii larger than the given threshold, the above two
techniques
will be combined.
For
simplification,
we will only consider the two special cases respectively in subsections \ref{subsec:large} and \ref{subsec:small}.

\subsection{Safety Verification of Interval Hybrid Systems With Large Radii 
}\label{subsec:large}
Let $\I\HH: \langle V, L, {\mathcal T},  { \Theta},  [\mathcal D], \Psi,  \ell_0\rangle$ be an interval hybrid system.
Suppose that the radii of the intervals in
the interval differential rules $[\mathcal D]$
 are all greater than a given threshold $\epsilon$,
say $\epsilon=0.1$. Then
some new parameters $u_1, \dots, u_t$ will be introduced to replace
the interval coefficients,
to convert $\I\HH$ into an uncertain hybrid system $\HH_{\uu}$
with $\uu =(u_1, \dots, u_t)$, for which Theorem~\ref{thm:location_inv}
can be extended to handle safety
verification.

Denote by $[\uu]=[\underline{\uu}, \overline{\uu}]\in\III\RR^t$ the
interval coefficient vector composed of all the interval
coefficients occurred in~$[\mathcal D]$, where
$\underline\uu=(\underline u_1, \dots,\underline u_t)$ and
$\overline\uu=(\overline u_1, \dots,\overline u_t)$.
To remove the intervals $[\uu]$ in $\I\HH$,
we introduce a vector
$\uu\in\RR^s$ of uncertainties with the constraints
$$\vartheta_i(\uu)=(u_i-\underline{u}_i)(\overline{u}_i-u_i)
\geq 0,\quad i=1,\dots, t.$$

For the uncertain
hybrid system $\HH_\uu$, we predetermine a template
$\varphi(\xx)=\sum_{\alpha} c_{\alpha}\xx^{\alpha}
$ of polynomial invariants
with the given degree $d$,
where~$\xx^{\alpha}=x_{1}^{\alpha_{1}}\cdots x_{n}^{\alpha_{n}}$,
$\alpha=(\alpha_{1},\ldots,\alpha_{n}) \in {\ZZ}_{\geq 0}^{n}$ with
$\sum_{i=1}^{n} \alpha_{i} \leq d$, and $c_{\alpha} \in \RR$ are
parameters.
For each location $\ell\in L$, we write $\varphi_{\ell}(\xx) = \cc_{\ell}^{T}\cdot
T(\xx)$, where  $T(\xx)$ is the (column) vector of all
terms in $x_1,\dots, x_n$ with total degree $\le d$, and $\cc_{\ell}
\in \RR^{\nu}$, with $\nu={n+d \choose n}$,  is the coefficient
vector of  $\varphi_{\ell}(\xx)$.
For clarity, we write $\varphi_{\ell}(\xx)$ as
$\varphi_{\ell}(\xx,\cc_\ell)$.
Similar to Theorem \ref{thm:location_inv}, the problem of computing
the invariants $\varphi_{\ell}(\xx)$ of the uncertain hybrid system
$\HH_\uu$ can be translated into the following problem
\begin{eqnarray}\label{problem_uncertain}
\left\{\begin{array}{l@{}l}
    &\text{find}\  \cc_{\ell}\in \RR^{\nu},\quad
\forall \ell\in L \\
                  & \text{s.t.}
                  \,\, \Theta \models \varphi_{\ell_0}(\xx,\cc_{\ell_0}) \geq 0,\\
                  & \quad\varphi_\ell(\xx,\cc_{\ell})\geq 0  \wedge g(\ell, \ell') \wedge \rho(\ell,\ell')\models \varphi_{\ell'}(\xx',\cc_{\ell'})\geq 0,  \\
                  &\quad \varphi_\ell(\xx,\cc_{\ell})\geq0 \wedge\Psi(\ell)\wedge \vartheta(\uu)\geq0\models \dot{\varphi}_\ell(\xx,\uu,\cc_{\ell}) > 0, \\
                  &\quad X_u(\ell)\models \varphi_\ell(\xx,\cc_{\ell}) < 0,
\end{array}\right.
\end{eqnarray}
where
$\dot{\varphi_{\ell}}(\xx,\uu,\cc_{\ell})=\sum_{i=1}^n\frac{\partial\varphi_{\ell}}{\partial
x_i} \cdot f_{\ell,i}(\xx,\uu)$.
%
Without loss of generality,
we consider a simpler form of (\ref{problem_uncertain}):
\begin{equation}\label{uncertain:reduced}
 \left \{
\begin{array}{l@{}l} \displaystyle
\text{find} & \quad  \cc \in \RR^{\nu} \\
 \text{s.t.} & \quad \,\,\,  \varphi_1(\xx,\cc)\geq0,   \\
 & \quad \,\,\,
 \varphi_{3}(\xx,\cc)\geq0\models\varphi_{2}(\xx,\cc)\geq0,\\
 & \quad \,\,\,
 \varphi_{5}(\xx,\uu,\cc)\geq0\models\varphi_{4}(\xx,\uu,\cc)\geq0,
\end{array}\right.
\end{equation}
where the coefficients of the polynomials $\varphi_{i}$'s are affine in $\cc$, for $i=1,\dots, 5$.
By Appendix B,
the problem (\ref{uncertain:reduced}) can be further transformed
into the following polynomial parametric optimization problem
  \begin{equation}\label{SOS:uncertain}
 \left \{
\begin{array}{l@{}l} \displaystyle
\text{find} & \,\,  \cc \in \RR^{\nu} \\
 \text{s.t.} & \,\,   \varphi_1(\xx,\cc)=
\mm_{1}(\xx)^{T} \cdot W^{[1]} \cdot  \mm_{1}(\xx),   \\
 & \,\,   \varphi_{2}(\xx,\cc)=
\mm_{2}(\xx)^{T}\cdot W^{[2]} \cdot \mm_{2}(\xx)
+(\mm_{3}(\xx)^{T}\cdot W^{[3]}\cdot \mm_{3}(\xx))\cdot
\varphi_{3}(\xx,\cc),\\
& \,\,   \varphi_{4}(\xx,\uu,\cc)= \mm_{4}(\xx,\uu)^{T}\cdot W^{[4]}
\cdot \mm_{4}(\xx,\uu) + (\mm_{5}(\xx,\uu)^{T}\cdot W^{[5]}\cdot
\mm_{5}(\xx,\uu))\cdot \varphi_{5}(\xx,\uu,\cc),
\\ & \,\,   W^{[i]} \succeq 0,  \quad i=1,\dots, 5,
\end{array}\right.
\end{equation}
which involves both LMI and BMI constraints. As stated in
\cite{yang2012exact}, a Matlab package PENBMI solver~\cite{KS2005},
which combines the (exterior) penalty and (interior) barrier method
with the augmented Lagrangian method, can be applied directly on the
BMI program, and alternatively, an iterative method can be applied
by fixing $W^{[5]}$ and $\cc$ alternatively, which leads to a
sequential convex LMI problem.




Since the SDP solvers in Matlab is running in fixed precision, the
above techniques yield numerical vector $\cc$ and numerical positive
semidefinite matrices $W^{[1]},\dots, W^{[5]}$, which satisfy the
constraints in~(\ref{SOS:uncertain}) {\em approximately}.
We will  apply the
symbolic-numeric
method proposed in \cite{yang2012exact} to obtain exact solutions to
(\ref{SOS:uncertain}). The idea is as follows. We first convert $W^{[3]}$ and $W^{[5]}$ to the
nearby rational positive semidefinite matrices $\widetilde{W}^{[3]}$ and
$\widetilde{W}^{[5]}$, respectively, by nonnegative truncated \ldl-decomposition,
in which all the diagonal entries of the corresponding diagonal
matrix are preserved to be nonnegative. Then, using modified
Newton refinement and rational vector recovery techniques, we can
recover the
rational vector $\tilde \cc$ and the rational positive semidefinite
matrices $\widetilde{W}^{[1]}$, $\widetilde{W}^{[2]}$,
$\widetilde{W}^{[4]}$  from the numerical $\cc, W^{[1]}, W^{[2]}, W^{[4]}$, respectively, such that the constraints in
(\ref{SOS:uncertain}) hold exactly. For more details, please refer to \cite{yang2012exact}.

\subsection{Safety Verification of Interval Hybrid systems with Small
Radii 
}\label{subsec:small}

In this subsection, we will consider
interval hybrid systems with small radii interval coefficients,
namely, the radii of the involved intervals are all smaller than the given
threshold $\epsilon$. For such interval hybrid systems, the method
described in subsection~\ref{subsec:large}
via introducing uncertainties
may suffer from high complexity
especially when solving the parametric
optimization problem (\ref{SOS:uncertain}).
%
Instead, we will consider how
to generate invariants
of $\I\HH$
by
determining
nonnegativity of
interval polynomials:
we first compute
candidate invariants with rational coefficients,
then
employ the interval computation method presented in
Section~\ref{sec:Nonnegativity} to certify that the candidate
invariants satisfy the conditions in Theorem~\ref{thm:inv_int}
exactly.

Suppose that $[D](\ell)$ of $\I \HH$ is
given by $\dot \xx=[\ff_{\ell}](\xx)$ for $\ell\in L$.
Choosing the midpoints of the interval coefficients of $[\ff_{\ell}](\xx)$ yields a mid-point vector
$\text{mid}\ff_{\ell}(\xx)$
and an associated general hybrid system $\HH$ with the vector field
$\dot\xx =\text{mid}\ff_\ell(\xx)$, for $\ell\in L$.
Then the symbolic-numeric technique in \cite{yang2012exact} can be used
to generate invariants of $\HH$ as follows.
Let us predetermine a polynomial template $\varphi_{\ell}(\xx)\ge 0$ of invariants of $\HH$ with $\deg \varphi_{\ell}(\xx)=d$.
By Theorem~\ref{thm:location_inv}, the problem of computing
$\varphi_{\ell}(\xx)$
can be translated into the following
problem
\begin{eqnarray}\label{problem_mid}
\left\{\begin{array}{l@{}l}
   & \text{find}\  \cc_{\ell}\in \RR^{\nu},\quad
\forall \ell\in L \\
                  & \text{s.t.}
                   \Theta \models \varphi_{\ell_0}(\xx,\cc_{\ell_0}) \geq 0,\\
                  & \quad\varphi_\ell(\xx,\cc_{\ell})\geq 0  \wedge g(\ell, \ell') \wedge \rho(\ell,\ell')\models \varphi_{\ell'}(\xx',\cc_{\ell'})\geq 0,  \\
                  &\quad \varphi_\ell(\xx,\cc_{\ell})\geq0 \wedge\Psi(\ell)\models \text{mid}\dot{\varphi}_\ell(\xx,\cc_{\ell}) > 0, \\
                  &\quad X_u(\ell)\models \varphi_\ell(\xx,\cc_{\ell}) < 0,
\end{array}\right.
\end{eqnarray}
where
$\text{mid}\dot{\varphi_{\ell}}(\xx,\cc_{\ell})=\sum_{i=1}^n\frac{\partial\varphi_{\ell}}{\partial
x_i} \cdot \text{mid}f_{\ell,i}(\xx)$. By use of BMI solving and
modified Newton refinement, we can obtain the refined numerical
solutions to (\ref{problem_mid}).
With the refined vector $\cc_{\ell}$ for $\ell\in L$, we
then apply rational vector recovery technique to obtain
a polynomial
${\varphi}_{\ell}(\xx,\tilde{\cc}_{\ell})$ with rational
coefficients.
Clearly, ${\varphi}_{\ell}(\xx,\tilde{\cc}_{\ell})$ can be seen as a candidate invariant of the interval hybrid system $\I \HH$.

In the following, we will determine whether
${\varphi}_{\ell}(\xx,\tilde{\cc}_{\ell})$ satisfies the
conditions of invariants of interval hybrid system $\I\HH$ in
Theorem \ref{thm:inv_int} exactly, i.e.,
\begin{eqnarray}\label{verification:interval}
\left\{\begin{array}{l@{}l}
                  & \Theta \models {\varphi}_{\ell_0}(\xx,\tilde{\cc}_{\ell_0}) \geq 0,\\
                  & {\varphi}_\ell(\xx,\tilde{\cc}_{\ell})\geq 0  \wedge g(\ell, \ell') \wedge \rho(\ell,\ell')\models {\varphi}_{\ell'}(\xx',\tilde{\cc}_{\ell'})\geq 0,  \\
                  & {\varphi}_\ell(\xx,\tilde{\cc}_{\ell})\geq0 \wedge\Psi(\ell)\models [\dot{{\varphi}}_\ell](\xx,\tilde{\cc}_{\ell}) > 0, \\
                  & X_u(\ell)\models {\varphi}_\ell(\xx,\tilde{\cc}_{\ell}) < 0,
\end{array}\right.
\end{eqnarray}
where
$[\dot{{\varphi}}_\ell](\xx,\tilde{\cc}_{\ell})=\sum_{i=1}^n\frac{\partial{\varphi}_{\ell}}{\partial
x_i} \cdot [f_{\ell,i}](\xx)$ is an interval polynomial.
Observing
in (\ref{verification:interval}),
all the constraints except the third one are exact constraints.
And the SOS-based method presented in subsection~\ref{subsec:large} can be used
to determine satisfiability of the exact constraints.
%
To handle the third constraint in (\ref{verification:interval}),
we now consider how to determine satisfiability of polynomial inequalities with interval coefficients.
More generally, we consider the following problem
\begin{equation}\label{interval:reduced}
  \psi_1(\xx,\tilde{\cc})\geq0\models[\psi_2](\xx,\tilde{\cc})\geq0,
\end{equation}
where $[\psi_2](\xx,\tilde{\cc})\in \III\RR[\xx]$.
Let $\text{mid}\psi(\xx,\tilde{\cc})\in[\psi_2](\xx,\tilde{\cc})$ be the mid-point function of $[\psi_2](\xx,\tilde{\cc})$. Then
BMI solver and modified Gauss-Newton refinement can yield the numerical positive semidefinite matrices
$W^{[1]}$ and $W^{[2]}$, which satisfy the following
condition approximately
\begin{equation}\label{interval:approx}
\left.\begin{array}{l@{}l} &\text{mid}\psi(\xx,\tilde{\cc})\approx
\mm_{2}(\xx)^{T}\cdot W^{[2]} \cdot \mm_{2}(\xx) +
(\mm_{1}(\xx)^{T}\cdot W^{[1]}\cdot \mm_{1}(\xx))\cdot
\psi_1(\xx,\tilde{\cc}).
\end{array}\right.
\end{equation}
Converting
${W}^{[1]}$ to a nearby rational positive
semidefinite matrix $\widetilde{W}^{[1]}$ by nonnegative truncated
\ldl-decomposition,
the condition
(\ref{interval:approx}) becomes
\begin{equation}\label{inter_evo}
\left.\begin{array}{l@{}l} &\text{mid}\psi(\xx,\tilde{\cc})
-(\mm_{1}(\xx)^{T}\cdot \widetilde{W}^{[1]}\cdot \mm_{1}(\xx))\cdot
\psi_1(\xx,\tilde{\cc})\approx \mm_{2}(\xx)^{T}\cdot W^{[2]} \cdot
\mm_{2}(\xx).
\end{array}\right.
\end{equation}
Let $[\tilde{\psi_2}](\xx, \tilde \cc)$ be an interval polynomial such that
$$[\tilde{\psi_2}](\xx,\tilde{\cc})=[\psi_2](\xx,\tilde{\cc})
-(\mm_{1}(\xx)^{T}\cdot \widetilde{W}^{[1]}\cdot \mm_{1}(\xx))\cdot \psi_1(\xx,\tilde{\cc}).$$
Since $\widetilde{W}^{[1]}\succeq 0$,
it suffices to prove satisfiability of (\ref{interval:reduced}) when $[\tilde{\psi}_2](\xx,\tilde{\cc})$ is nonnegative.
Remark that (\ref{inter_evo}) is an approximate SOS decomposition of
$[\tilde{\psi_2}](\xx,\tilde{\cc})$. The nonnegativity of
$[\psi_{2}]$
can be verified by computing the corresponding
interval matrix $[W_{2}]$, either using the method in
subsection~\ref{subsec:fullrank} if $W^{[2]}$ is of full rank, or
by proving existences of real roots of the interval
nonlinear system, as explained in subsection~\ref{subsec:singular}.

\subsection{Experiments}\label{algorithm}

In the following, some examples will be given to illustrate our
method on safety verification of interval hybrid systems.

\begin{example}\label{ex2}
Consider the classical two-dimensional system given in
\cite{khalil1996nonlinear,prajna2007safety}, whose coefficients are
approximated and described by the following intervals
\begin{equation*}
\left[\begin{array}{ccc}\dot{x}_1\\\dot{x}_2\end{array}\right]=\left[
\begin{array}{ccc}[0.99,1.01]x_2 \\ -[0.96,1.04]x_1+[0.32,0.347]x_1^3-[0.98,1.02]x_2\end{array}\right].
\end{equation*}
We will verify that all trajectories of the
system starting from the initial set
$$\Theta=\{(x_1,x_2)\in \RR^2:
(x_1-1.5)^2+x_2^2\leq0.25\}$$
 will never enter the
unsafe region
$$X_u=\{(x_1,x_2)\in \RR^2: (x_1+1)^2+(x_2+1)^2\leq0.16 \}.$$

%

Set the threshold $\epsilon=0.1$
Clearly, all the radii of involved intervals are less than this
threshold.
Applying the  method in subsection \ref{subsec:small},
we obtain the following verified invariant with
rational coefficients
$$
   \widetilde{\varphi}(\xx)={\frac {151}{99}}+{\frac {152}{99}}\,x_1+{\frac {62}{33}}
\,x_2+{\frac {106}{99}}\,x_1x_2+{\frac {4}{9}}\,x_1^{2},
$$
which guarantees the safety of the original system. $\hfill \Box$
\end{example}

\begin{example}
Consider a Moore-Greitzer model of a jet engine with stabilizing
feedback operating in the no-stall mode \cite{aylward2008stability}.
In this model, the origin is translated to a desired no-stall
equilibrium. The dynamic system takes the following form:
\begin{eqnarray*}\label{ex_1}
\left\{\begin{array}{l@{}l}
   \dot{x}_1={[-1.1,-0.9]}x_2-\frac{3}{2}x_1^2-\frac{1}{2}x_1^3, \\
   \dot{x}_2={[2.98,3.02]}x_1-x_2.
   \end{array}\right.
\end{eqnarray*}
The problem is to verify that all trajectories of the system
starting from the initial set
$$\Theta=\{(x_1,x_2)\in \RR^2:
(x_1-1)^2+x_2^2\leq 0.04\}
$$
will never reach the unsafe set
$$X_u=\{(x_1,x_2)\in \RR^2:(x_1+1.8)^2+x_2^2\leq 0.16\}.$$

Set the threshold $\epsilon$ of radii to be $0.1$. Then a new
uncertainty $u$ is introduced to replace the interval $[-1.1,-0.9]$.
Combine the methods in subsections \ref{subsec:large} and \ref{subsec:small} to deal with the uncertain interval
system, and we obtain the following verified invariant with rational
coefficients
$$   \widetilde{\varphi}(x_1,x_2)={\frac {2231}{328}}+{\frac
{652}{123}}x_1+{\frac {274}{123}}x_2-{ \frac {46}{41}}x_1^{2}+{\frac
{10}{41}} x_1x_2+{\frac {1649}{984}}x_2^{2},
$$
which guarantees the safety of the original system. $\hfill \Box$
\end{example}

\begin{example}
Figure \ref{hybrid2} gives a predator-prey hybrid system \cite{ratschan2005safety}
with interval coefficients:
\begin{eqnarray*}\label{ex_2}
f_1(\xx)=f_2(\xx)=\left[
\begin{array}{lll}-x_1+[0.99,1.01]x_1x_2 \\
{[0.875,1.2]}x_2-x_1x_2\end{array}\right].
\end{eqnarray*}
Suppose the system starts in
location $\ell_1$ with an initial state in
$$\Theta=\{(x_1,x_2)\in \RR^2: (x_1- 0.8)^2 +(x_2- 0.2)^2 \leq
0.01\}.$$
\begin{figure}[ht]\label{hybrid2}
\centering
\includegraphics[width=0.65\textwidth]{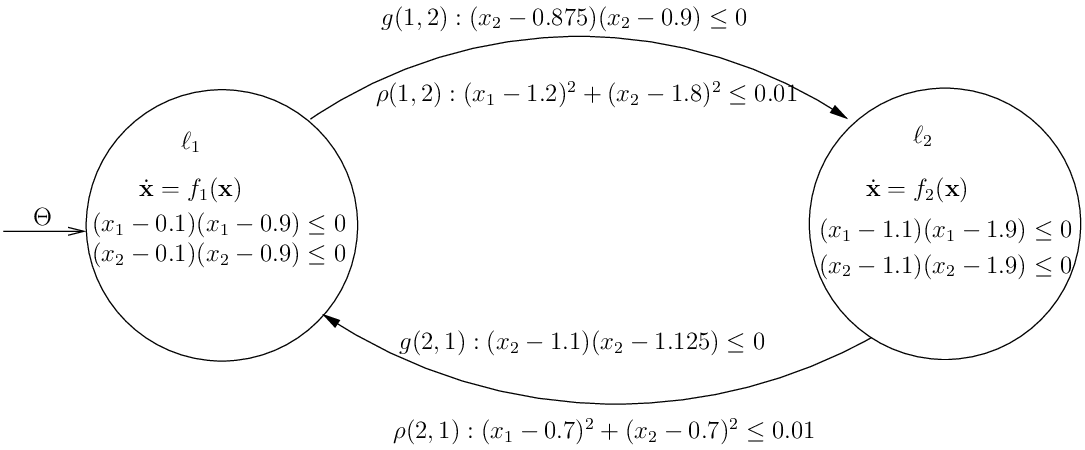} \caption{\small Hybrid system
of Example 4}\label{C5Fig1}
\end{figure}
We want to  verify that the system never reach the states
in
$$X_u(\ell_1)=\{(x_1,x_2)\in \RR^2: 0.8 \leq x_1 \leq 0.9 \wedge 0.8
\leq x_2 \leq 0.9 \}.$$

Set the threshold $\epsilon$ of radii to be  $0.1$. Then a new
uncertainty $u$ is introduced to replace the interval $[0.875,1.2]$.
Applying the above method 
on the resulting uncertain interval hybrid system $\I\HH_u$, we
obtain the following verified invariants with rational coefficients
\begin{eqnarray*}
    &&\widetilde{\varphi}_{1}(x_1,x_2)=-{\frac {411}{995}}+{\frac {346}{995}}\,{x_1}+{\frac {397}{995}}\,{
x_2}-{\frac {49}{199}}\,x_2^{2}, \\
    &&\widetilde{\varphi}_{2}(x_1,x_2)={\frac {556}{995}}-{\frac {151}{199}}\,{x_1}-{\frac {986}{995}}\,{
x_2}+{\frac {22}{995}}\,x_2^{2},
\end{eqnarray*}
for locations $\ell_1$ and $\ell_2$, respectively, which ensures the safety of the
original hybrid system. $\hfill \Box$
\end{example}

\section{Safety Verification of Non-polynomial Hybrid system}\label{nonpolynomial}
As an application of the method in Section~\ref{sect:Interval} for
safety verification for interval hybrid systems, we will consider
how to verify safety of {\em non-polynomial hybrid systems}.

Let
$\HH: \langle V, L, {\mathcal T}, { \Theta}, \mathcal D,
\Psi,  \ell_0\rangle
$
be a hybrid system where the initial condition $\Theta$,
location invariants~$\Psi(\ell)$, the guard condition and reset
relation in each transition $\tau \in \mathcal T$ are semialgebraic
sets, whereas the
continuous systems in the
differential rules $\mathcal D(\ell)$,
contain
some non-polynomial terms
in $\xx$.
%
For such a
non-polynomial hybrid system $\HH$, we will
first transform it into an uncertain interval hybrid system $\I\HH$
through polynomial approximation on the non-polynomial terms, such
that $\HH$ is within $\I\HH$.
This implies that the safety of $\I\HH$ suffices to guarantee the
safety of $\HH$, and then the method in Section~\ref{sect:Interval} can
be applied to the former problem.

%

Assume that the location invariant
$\Psi(\ell)$ is a compact set for each location $\ell$.
Consider the continuous dynamics of a hybrid system $\bf H$ at location $\ell$:
\begin{equation}\label{eq:nonpolynomial}
  \dot{x}_i=f_i(\xx)=f_{i0}(\xx)+\sum_{j=1}^s f_{ij}(\xx)\phi_{ij}(\xx),\,\, i=1,\dots, n,
\end{equation}
where $\xx$ takes values in $\Psi(\ell)\subseteq\RR^n$,
$f_{ij}(\xx)$ are polynomials for~$j=0,1,\dots,s$, and
$\phi_{ij}(\xx)$ are non-polynomials for $j=1,\dots, s$.
We will approximate the functions
$\phi_{ij}(\xx)$ with polynomials
$g_{ij}(\xx)$ for~$i=1,\dots, n$ and $j=1,\dots, s$.
Let $\mu_{ij}$ be the bound of $|\phi_{ij}(\xx)-g_{ij}(\xx)|$ for all
$\xx\in\Psi$, namely,
\begin{equation}\label{non_polynomial_bound}
|\phi_{ij}(\xx)-g_{ij}(\xx)|\leq \mu_{ij}, \quad\mbox{for all } x
\in \Psi(\ell).
\end{equation}
Making use of the relation~(\ref{non_polynomial_bound}) for each
location $\ell\in L$, we can construct an interval polynomial hybrid
system $\I\HH: \langle V, L, {\mathcal T}, { \Theta}, [\mathcal D],
\Psi,  \ell_0\rangle$,
where the interval differential rule $[\mathcal D(\ell)]$ given by
\begin{equation}\label{EQ:polynomial}
  \dot{x}_i=[f_i](\xx)=f_{i0}(\xx)+\sum_{j=1}^s f_{ij}(\xx)(g_{ij}(\xx)+[-\mu_{ij},\mu_{ij}]),\,i=1,\dots, n,
\end{equation}
enclosures the non-polynomial system (\ref{eq:nonpolynomial}) in
$\HH$,
that is, $f_{i}(\xx) \in [f_{i}](\xx)$ for all $\xx\in \Psi(\ell).$


The key point of the above idea is to compute
an approximate polynomial and the associated bound  for the given non-polynomial function. For
a non-polynomial function $\phi(\xx)$ with $\xx \in \Psi(\ell)$, we will compute the approximate polynomial $g(\xx) \in
\RR[\xx]$ with a verified bound $\mu \in \RR_{+}$, such that
   $$|\phi(\xx)-g(\xx)|<\mu, \forall x \in \Psi(\ell),$$
and the bound $\mu$ is as small as possible.

A classic method of polynomial approximation is Taylor expansion. In
this paper, to obtain a tighter error bound, multivariate polynomial
interpolation\cite{gasca2000history} is applied to compute an
approximate polynomial with the error bound.
Furthermore, the
technique of oversampling is explored to get better approximate
polynomials, i.e., the number of the interpolation points is
greater than that of the terms of the target polynomial
$g(\xx)$.
Given the interpolation points,
the approximate polynomial $g(\xx)$ can be obtained by solving a least squares problem.
Specifically, predetermine a polynomial template of $g(\xx)$ with a
given degree $r$:
  \begin{equation}\label{g_form}
     g(\xx) = \cc^{T}\cdot T(\xx),
    \end{equation}
where $T(\xx)$ is the (column) vector consisting of all terms in
$x_1,\dots, x_n$ with total degree $\le r$, and $\cc \in \RR^{\nu}$,
with $\nu={n+r \choose n}$, is the coefficient vector of $g(\xx)$.
We then construct a mesh $M$ on $\Psi(\ell)$ with a small spacing
$s\in\RR_+$, and compute $y_j=\phi(\vv_j) \in \RR$ for $1\leq j\leq
m$ at mesh points $\{\vv_1, \vv_2, ..., \vv_m\}$.
Let the coefficient vector $\cc$ of $g(\xx)$ be
unknowns. We can construct a linear system
\begin{equation}\label{system_overdetermined}
   A\cdot\cc=\yy,
\end{equation}
where $A=(T(\vv_1)^T,T(\vv_2)^T,...,T(\vv_m)^T)^T$ is of size ${m\times\nu}$ with $m>\nu$. By
solving the above overdetermined system, we obtain $g(\xx,\cc)$ as the approximation of $\phi(\xx)$ with $x\in \Psi(\ell)$.
Having
$g(\xx,\cc)$, one will compute  the verified  error bound $\mu$,
namely,
  $|\phi(\xx)-g(\xx,\cc)|<\mu, \forall \xx \in \Psi(\ell).$

\begin{lemma}\label{thm:Zeng}
\cite[Theorem 3]{zeng2010} Let $K\subset\RR^n$ be a convex
polyhedron, and $V_1,V_2,...V_m$ and $d$ be the vertices and
diameter of $K$ respectively. Suppose $\psi: K\rightarrow\RR$ is a
continuous and differential function on $K$, then for all
$a_1,a_2,...a_m\in\RR_+$ such that $ a_1+a_2+...+a_m=1$, we have
$$
  |\psi(\xx)-(a_1\psi(V_1)+a_2\psi(V_2)+...+a_m\psi(V_m))|\leq \frac{n}{n+1}\beta d,
$$
where
$\beta=\sup_{\xx\in K} \|\bigtriangledown \psi(\xx)\|.$
\end{lemma}

For the error function $r(\xx)=\phi(\xx)-g(\xx,{\cc})$, we will
estimate its bound with $\xx$ in the mesh $M$ by the following
theorem.

\begin{theorem}\label{thm:bound_error}
Suppose that $s$ and $\{\vv_1, \vv_2, ..., \vv_m\}$ are the mesh
spacing and mesh points of $M$, respectively. Let
$\mu_0=\max\{r(\vv_1),r(\vv_2),...,r(\vv_m)\}$, and
$\beta'=\sup_{\xx\in M} \|\bigtriangledown r(\xx)\|,$ then for all
$\xx\in M$,
$$
  |r(\xx)|\leq \frac{n}{n+1}\beta's+\mu_0. 
$$
\end{theorem}

\begin{proof}
We know that $r(\xx)$ is a continuous and differential function on
$M$.
Thus, according to Lemma \ref{thm:Zeng}, for all
$a_1,a_2,...a_m\in\RR_+$ such that $ a_1+a_2+...+a_m=1$,
$$|r(\xx)-(a_1r(\vv_1)+a_2r(\vv_2)+...+a_mr(\vv_m))|\leq
\frac{n}{n+1}\beta's.
$$
Then, we have
\begin{equation*}
\begin{split}
|r(\xx)|&\leq\frac{n}{n+1}\beta's+|(a_1r(\vv_1)+a_2r(\vv_2)+...+a_mr(\vv_m))|\leq\frac{n}{n+1}\beta's+\mu_0.
\end{split}
\end{equation*}
\end{proof}


\begin{example}\label{ex4}
Consider the function $\phi(x)=e^x$ with $\Psi:-2\leq x\leq 2$. We
want to compute a polynomial $g(x)$  and the associated verified error bound $\mu$ such that
 $$|\phi(x)-g(x)|<\mu, \,\, -2\leq x\leq 2.$$
First, we construct a mesh $M$ on $\Psi$ with the spacing
$s=\frac{1}{4}$. For a polynomial of the form
$g(x,\cc)=c_0+c_1\,x+c_2\,x^2+c_3\,x^3$, it is easy to find an
approximate polynomial
$$g(\xx,\hat{\cc})=0.9173+0.9562x+0.6797x^2+0.2117x^3.$$ According to
Theorem \ref{thm:bound_error}, we can also compute the error bound
$\mu=0.2937$.
\begin{figure}
\centering
\includegraphics[width=0.35\textwidth,angle=270]{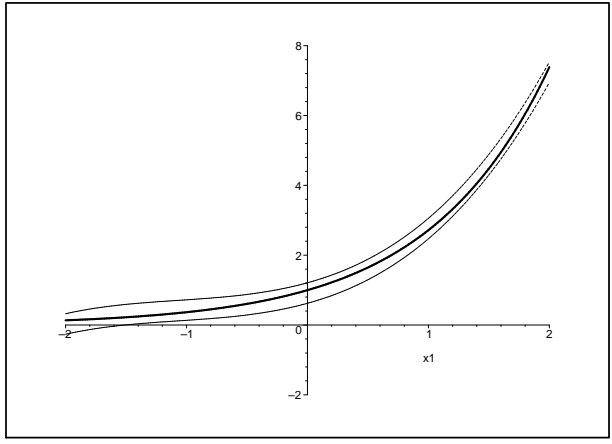} \caption{\small Approximate $e^x,\,-2\leq x\leq 2$ by
$g(x,\hat{\cc})+[-\mu,\mu]$(solid line: $e^x$, dot line:
$g(x,\hat{\cc})\pm\mu$).}
\end{figure}
The results are as shown in Figure 2. $\hfill \Box$
\end{example}

Stated as above,
once we obtain an interval
polynomial hybrid system $\I \HH$ from $\HH$ through polynomial approximation such that $\HH$ is within $\I\HH$,
the method in Section \ref{sect:Interval} can be used to verify
safety of $\I \HH$, which ensures  safety of $\HH$.
The following example is presented to illustrate our method for safety
verification of a non-polynomial hybrid system.

\begin{example}\label{ex:two-tank}
Consider the following two-tanks hybrid system
\cite{ratschan2007safety} depicted in Figure \ref{hybrid1} with
$$f_1(\xx)=
\begin{bmatrix}
1-\sqrt{x_1}  \\
\sqrt{x_1}-\sqrt{x_2}
\end{bmatrix}
, \quad f_2(\xx)=
\begin{bmatrix}
1-\sqrt{x_1-x_2+1} \\
\sqrt{x_1-x_2+1}-\sqrt{x_2}
\end{bmatrix}
,$$ where $\xx=(x_1, x_2)$ denotes the liquid levels. In
\cite{ratschan2007safety}, the authors verified that the system
starting in location $\ell_1$ with an initial state in
$$\Theta=\{(x_1,x_2)\in \RR^2: (x_1-5.5)^2+(x_2-0.25)^2 \leq 0.0625\}$$
never reach the states of
\begin{equation*}
\begin{split}
X_u(\ell_1)=\{(x_1,x_2)\in \RR^2: (x_1-4.25)^2+(x_2-0.25)^2 \leq
0.0625 \}.
\end{split}
\end{equation*}
\begin{figure}
\centering
\includegraphics[width=0.65\textwidth]{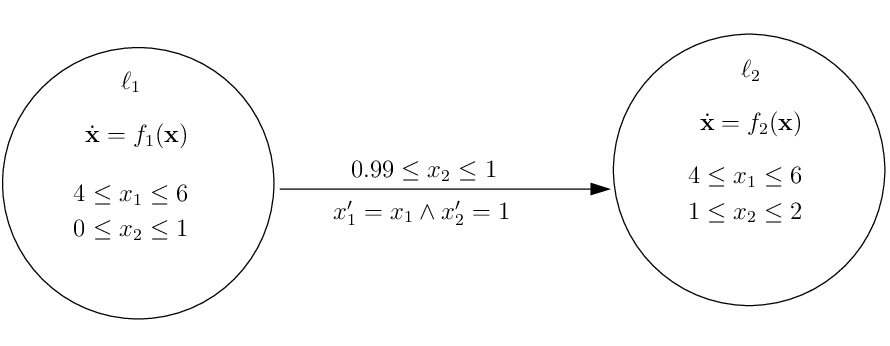} \caption{\small Hybrid system
of Example~\ref{ex:two-tank} \label{hybrid1}}
\end{figure}
Here, we enlarge both radii of initial and unsafe regions to $0.49$,
that is,
$$\Theta=\{(x_1,x_2)\in \RR^2: (x_1-5.5)^2+(x_2-0.25)^2 \leq 0.2401\}$$
and
\begin{equation*}
\begin{split}
X_u(\ell_1)=\{(x_1,x_2)\in \RR^2: (x_1-4.25)^2+(x_2-0.25)^2 \leq
0.2401 \},
\end{split}
\end{equation*}
and consider again safety verification of the given system.
We first compute an interval polynomial system given by
$[{f}_1](\xx)$ and $[{f}_2](\xx)$ to enclosure the original system
where
\begin{equation*}
\begin{split} &[{f}_1](\xx)=\begin{bmatrix}
[0.1658,0.173]-0.3377x_1+0.0114x_1^2 \\
[0.6465 ,0.8615 ]+0.3377x_1-1.7115x_2-0.0114x_1^2+0.8241x_2^2
\end{bmatrix}
\end{split}
\end{equation*}
and
\begin{equation*}
\begin{split} &[{f}_2](\xx)
=\begin{bmatrix}
-[0.1204,0.132]-0.3316x_1+0.3319x_2+0.0135x_1^2-0.0269x_1x_2+0.0137x_2^2 \\
[0.6716,0.6898]+0.3316x_1-0.9576x_2-0.0135x_1^2+0.0269x_1x_2+0.0572x_2^2
\end{bmatrix}.
\end{split}
\end{equation*}
We obtain the following verified invariants with rational
coefficients
\begin{equation*}
\begin{split}
&\widetilde{\varphi}_{1}(\xx)=-{\frac {1069}{994}}-{\frac
{145}{142}}{ x_1}-{\frac {367}{497}} { x_2}+{\frac
{121}{497}}x_1^{2}+{\frac {160}{497}}{ x_1}{ x_2}+{\frac {242}{497}}
x_2^{2}, \\
&\widetilde{\varphi}_{2}(\xx)={\frac {9621}{994}}-{\frac
{20}{71}}{x_1}+{\frac {899}{497}}{ x_2}+{\frac
{989}{994}}x_1^{2}+{\frac {349}{497}}{x_1}{x_2}-{\frac
{1487}{994}}x_2^{2},
\end{split}
\end{equation*}
which satisfy the conditions in Theorem~\ref{thm:inv_int} exactly.
Therefore,
the safety of the original hybrid system is verified. $\hfill \Box$
\end{example}

The above approach can be easily extended to the case of {\em
uncertain non-polynomial hybrid systems}, by which we mean the
continuous dynamics at each location $\ell$ are given by uncertain non-polynomial systems of the form
\begin{equation}\label{eq:nonpolynomial2}
  \dot{x}_i=f_i(\xx,\theta)=f_{i0}(\xx,\theta)+\sum_{j=1}^s f_{ij}(\xx,\theta)\phi_{ij}(\xx),\quad\mbox{for } 1\leq i\leq n,
\end{equation}
where $\theta\in\Phi\subseteq\RR^t$ is a vector of uncertainty. The
following example demonstrates how to apply the above approach to
verify
 safety of an uncertain non-polynomial system.

\begin{example}
Consider an uncertain non-polynomial system given in
\cite{chesi2009estimating}:
\begin{eqnarray*}\label{ex_3}
\left\{\begin{array}{l@{}l}
   \dot{x}_1=-x_1+x_2+\frac{1}{2}(e^x_1-1), \\
   \dot{x}_2=-x_1-x_2+\theta x_1x_2+x_1\cos{x_1},
   \end{array}\right.
\end{eqnarray*}
for $-2\leq x_1, x_2\leq 2$ and $0.98\leq \theta\leq 1.2$. This system
starts with an initial state in
$$\Theta=\{(x_1,x_2)\in \RR^2: x_1^2+x_2^2 \leq 0.25\},$$
and we want to verify that the system never reach the states of
$$X_u=\{(x_1,x_2)\in \RR^2: (x_1+1.5)^2+(x_2+1.5)^2\leq0.16 \}.$$
To prove the safety of this non-polynomial system, we first compute
interval polynomials to approximate the non-polynomial terms
$e^{x_1}$ and $\cos x_1$ occurred in this system. Based on the above
techniques, we obtain the following interval polynomial system
\begin{equation*}\label{ex5'}
\begin{split}
   &\dot{x}_1=[-0.1882,0.1055]-0.5219x_1+x_2+0.33985x_1^2+0.10585x_1^3, \\
   &\dot{x}_2=[-0.2067,0.0875]x_1-x_2+\theta x_1x_2-0.3594x_1^3,
   \end{split}
\end{equation*}
which enclosures the original system.
For the above interval hybrid system, we obtain the following
verified invariant with rational coefficients
\begin{equation*}
   \widetilde{\varphi}(\xx)={\frac {343}{32}}+{\frac {31}{6}}\,{x_1}+{\frac {25}{48}}\,{x_2}
-{\frac {49}{32}}x_1^{2}-{ \frac {17}{48}}\,{x_1}\,{x_2} -{\frac
{55}{32}}x_2^{2},
\end{equation*}
which guarantees the safety of the original system.
\end{example}

\section{Conclusion}\label{sect:conclusion}
In this paper, a hybrid symbolic-numeric method, based on SOS
relaxation and interval arithmetic certification, is proposed to
generate exact inequality invariants for safety verification of
interval hybrid systems. As an application,
one approach is provided to verify safety
property of non-polynomial hybrid systems. More precisely,
we apply a rigorous polynomial approximation method to compute an
interval polynomial system, which contains the non-polynomial
system, and then compute the exact invariant of the corresponding
interval polynomial system to verify the safety property of the
original system. Experiments on the benchmark hybrid systems
illustrate the efficiency of our algorithm.


\section*{Appendix}


\subsection*{A. Interval Arithmetic}\label{appendix B}

Interval arithmetic \cite{alefeld1983introduction} has been designed
for automatically
tackling roundoff errors of numerical computations.
In this subsection, some notions
about interval arithmetic are presented.

Denote the closed intervals by $[a],[b]$, etc. By convention,  the left and right endpoints of a closed
interval $[a]$ are represented by $\underline a$ and $\overline a$,
respectively, i.e.,
$$[a]=\{x\in \RR:\, \underline a \le x\le \overline a\}$$
with $\underline{a}=\inf[a]$ and $\overline{a}=\sup[a]$.
Any real number $a$ can also be regarded as an interval by
identifying
$a
$ with the ``point interval"
$[a]$ with $\underline a =\overline a =a$.
Such point intervals are also called {\em degenerate} intervals. Let
$$\text{mid}([a]):=\frac{1}{2}(\underline{a}+\overline{a})\quad\mbox{
and}\quad \text{rad}[a]:=\frac{1}{2}(\overline{a}-\underline{a})$$ be
the midpoint and radius of the closed interval $[a]$, respectively.
Clearly, an interval
can also be represented by its midpoint and radius. The set of all
intervals
over $\RR$ is denoted by $\III\RR$. The arithmetic operations
$+,-,*,\div$ can be extended from $\RR$ to $\III\RR$ in the usual
set theoretic sense, and the bounds of the resulting intervals can
be computed from the bounds of the operands, see
\cite{alefeld1983introduction} for details.

By $\III\RR^n$ and $\III\RR^{m\times n}$
we denote the sets of real $n$-dimensional vectors and $m\times n$
matrices
over $\III\RR$, respectively.
Elements of $\III\RR^n$ are called interval vectors and denoted by
$[\mathbf{a}], [\bb]$ and etc, and elements of $\III\RR^{m\times n}$
are called interval matrices and denoted by $[A], [B]$ and etc.
Remark that interval vectors (resp.\ interval matrices) are sets of
vectors (resp.\ matrices). For interval vectors and matrices, the
notions of  midpoints and radius, and the arithmetic operations are
defined
componentwise.

By an {\em interval} linear system, we mean a system of the form
\begin{equation}\label{int:linear}
[A]\xx=[\bb],
\end{equation}
where $[A]\in\III\RR^{n\times n}$, $[\bb]\in\III\RR^n$ and
$\xx=(x_1, \dots, x_n)^T$ is a column vector of $n$ unknowns. The
set
$$
   \Sigma=\{\xx\in\RR^n:\, A\xx=\bb \text{ for some } A\in [A], \bb\in[\bb]\}
$$
is called the {\em solution set} of the interval system
(\ref{int:linear}).
Many efficient algorithms are available in \cite{hansen1992bounding,
rohn1995computing, rump1992solution, rump2010verification}
for obtaining guaranteed inclusions \cite{rump2010verification}
for the solution set $\Sigma$.

Let $\ff: \RR^n\rightarrow\RR^n$ be a continuously differentiable
function. Replacing the real vector $\xx$ by an intervector
$[\xx]\in\III\RR^n$ we thus obtain an interval extension $[\ff]$ of
$\ff$. By the inclusion property of interval arithmetic, the range
of $\ff$ over an interval is contained in its interval extension,
i.e. $\{\ff(\xx):\, \xx\in[\xx]\}\subseteq[\ff]([\xx])$.
%
To determine existence of solutions to
the nonlinear system $\ff(\xx)=0$, we will use the Krawczyk
operators \cite{krawczyk1969newton}
based on Browder fixed points, which is defined as follows. Assume
that $[\xx]\in\III\RR^n$ is an interval set satisfying $\hat{\xx}\in
[\xx]$, and $C\in\RR^{n\times n}$. The Krawczyk operator is defined
as follows
$$
   K(\hat{\xx}, [\xx],\ff)=\hat{\xx}-C\ff({\hat{\xx}})+(I-C[\ff']([\xx]))([\xx]-\hat{\xx}).
$$
In practical computation, $C$ is usually chosen to be near the
inverse of the Jacobian $\ff'(\hat{\xx})$.
\begin{theorem}\label{thm:Ky} \cite{rump2010verification}
   Under the above assumptions, if
$$
  K(\hat{\xx}, [\xx], \ff)\subset \text{int}([\xx]),
$$
where $\text{int}([\xx])$ is the topological interior of $[\xx]$,
then there exists a unique $\xx^*\in K(\hat{\xx}, [\xx], \ff)$ such
that $\ff(\xx^*)=0$.
\end{theorem}

INTLAB is a MATLAB toolbox \cite{rump1999intlab}, which consists of
interval calculations,
and interval arithmetic for vectors and matrices.
Many interval operations
in this paper are implemented in MATLAB that uses the INTLAB package
supporting rigorous real interval standard functions and
interval least squares problem.

\subsection*{B. Sum of Squares Relaxation}\label{subsect:sos}
We give a brief review on SOS optimization. More details will be
found in \cite{parrilo2000structured}.
Recall that a sufficient condition for determining
$\psi(\xx)\in\RR[\xx]$ to be positive semidefinite
is that there exists an SOS decomposition of $\psi(\xx)$:
   \begin{equation}\label{matrix_form}
    \psi(\xx)=\sum_{i=1}^s {f_{i}}^{2}(\xx), \quad \text{
   with} \,\, f_{i}(\xx)\in\RR[\xx],
   \end{equation}
or, equivalently,
$\psi(\xx)$ can be represented in the Gram matrix form
   $$
 \psi(\xx)=m(\xx)^{T}\cdot W \cdot m(\xx),
$$
where~$W$ is a real symmetric and positive semidefinite matrix over
$\RR$,
 and $m(\xx)$ is a vector of
 all monomials
 in~$\RR[\xx]$ with degree $\leq \frac{1}{2}
\deg r(\xx)$.
Therefore the SOS program (\ref{matrix_form}) can be further
converted into the following Semidefinite programming (SDP) problem
\begin{equation}\label{SOSW}
\left.\begin{array}{l@{}l}
\displaystyle \inf_{W} \ & \text{Trace}(W) \\
\text{s.\ t.} \ &  \psi(\xx)=m(\xx)^T \cdot W \cdot  m(\xx) \smallskip\\
& W \succeq 0, W^T = W,
\end{array}\right\}
\end{equation}
where $\text{Trace}(W)$ acts as a dummy objective function that is
commonly used in SDP for optimization problem with no objective
functions. Many Matlab packages of SDP solvers, such as SOSTOOLS
\cite{sostools}, YALMIP \cite{yalmip}, and SeDuMi \cite{Sturm99},
are available to solve
the SDP problem~(\ref{SOSW}) efficiently.

The SOS programs have many applications,
for example, in determining the nonnegativity of a multivariate
polynomial over a semialgebraic set.
Consider the problem of verifying whether the implication
\begin{equation}\label{Implication}
\bigwedge_{i=1}^m(p_i(\xx)\ge 0)
\implies
 q(\xx)\ge 0
\end{equation}
holds, where $p_i(\xx)\in \RR[\xx]$ for~$1\le i\le m$ and $q(\xx)\in
\RR[\xx]$.
According to Stengle's Positivstellensatz, Schm\"{u}dgen's
Positivstellensatz or Putinar's Positivstellensatz
\cite{bochnak1998real}, if there exist SOS polynomials
$\sigma_{i}\in\RR[\xx]$ for $i=0,...,m$, such that
$$ q(\xx)=
\sigma_{0}(\xx)+\sum_{i=1}^m\sigma_i(\xx)p_i(\xx), $$ then the
assertion
(\ref{Implication}) holds. Therefore, the existence of SOS
representations provides a sufficient condition for determining the
nonnegativity of $q(\xx)$ over $\{\xx\in \RR^n : \bigwedge_{i=1}^m
p_i(\xx)\ge 0\}$.

\subsection*{C. Existence of Real Roots for
Underdetermined Interval Nonlinear Systems}


Consider a nonlinear system
\begin{equation}\label{interval:underdetermined}
   F(\qq)-[\vv]=0.
\end{equation}
where $F: \RR^r\rightarrow\RR^s$ a continuously differentiable
function with $r>s$, $r=\text{Dim}(\qq)$, and $[\vv]\in\III\RR^s$.
To determine the existence of solution to system
(\ref{interval:underdetermined}), we present two methods as follows.
The idea of the first method is to transform the underdetermined interval system into the corresponding
interval square nonlinear system by fixing some variables as constants,
and then generalize Theorem~\ref{thm:Ky} to verify the existence of real roots
for this square nonlinear system,
while
the second method is to generalize the method
in~\cite{chen2006existence} to
solve the interval underdetermined system
(\ref{int:underdetermined1}).


Firstly, suppose that $\hat{\qq}$ is an approximate solution of
(\ref{interval:underdetermined}). Here we assume that the Jacobian
matrix  $F'(\qq)$ at $\hat{\qq}$ is of full row rank. Column
pivoting $QR$-decomposition for $F'(\qq)$ is applied to choose an
index set $B=\{k_1,k_2,...,k_s\}$ such that
$F'_{B}({\hat{\qq}})\in\RR^{s\times s}$ is nonsingular, that is,
$$
 F'({\hat{\qq}})\,P^T=Q\,[\,R\,|\,*\,]\quad\mbox{with } P\in\RR^{r\times r}, \, Q,R\in\RR^{s\times
 s},
$$
where $P$ is a permutation matrix, $Q$ is orthogonal and $R$ is
upper triangular. The permutation $P$ arises from a greedy strategy
to obtain maximum diagonal elements in $R$. Then, the set $B$ can be
taken as those components which are permuted to the first $s$
positions by $P$. Thus, $\qq$ can be separated into two parts
$\qq=(\qq_B,\qq_N)$, where $N=\{1,2,...r\}/B$. Similar to the
partition of $\qq$, we have
${\hat{\qq}}=({\hat{\qq}}_B,{\hat{\qq}}_N)$. By use of the
evaluations $\qq_N=\hat{\qq}_N$,  (\ref{interval:underdetermined})
becomes the following interval square system
\begin{equation}\label{eq:square_int}
  G(\qq_B)-[\tilde{\vv}]=0,
\end{equation}
where $G(\qq_B)=F(\qq_B,\hat{\qq}_N)-\cc$, and
$[\tilde{\vv}]=\cc+[\vv]$, and $\cc$ is the constant vector of
$F(\qq_B,\hat{\qq}_N)$, i.e., $\cc=F(0,\hat{\qq}_N)$.

Observing in (\ref{eq:square_int}), 
the interval
coefficients only occur in the constant vector $[\tilde{\vv}]$, and
$G(\qq_B)$ is a
real function from $\RR^{s}$ to $\RR^{s}$,
meaning that the Jacobian matrix of (\ref{eq:square_int}) is the same
as one exact square system $G(\qq_B)-\vv=0$, where $\vv$ is
a
vector chosen randomly. Taking advantage of this property, it is
easy to generalize Theorem~\ref{thm:Ky} to verify the existence of
real roots for (\ref{eq:square_int}).

\begin{theorem}\label{cor:square}
Consider the system (\ref{eq:square_int}). Let $[\qq_B]\in\III\RR^s$
be
such that $\hat{\qq}_B\in[\qq_B]$, and $C\in\RR^{s\times s}$.
If
\begin{equation}\label{Ky2}
\begin{split}
   &K({\hat{\qq}}_B, [\qq_B], G-[\tilde{\vv}])=\hat{\qq}_B-C(G({\hat{\qq}}_B)-[\tilde{\vv}])+(I-CG'([\qq_B]))([\qq_B]-{\hat{\qq}}_B)\,\subset\,
\text{int}([\qq_B]),
\end{split}
\end{equation}
then there is a unique root ${\qq}^*_B$ of (\ref{eq:square_int}) in
$[\qq_B]$ for each $\vv\in[\tilde{\vv}]$.
\end{theorem}
\begin{proof}
If (\ref{Ky2}) holds, then we have, for each $\vv\in[\tilde{\vv}]$,
\begin{equation*}
\begin{split} &K({\hat{\qq}}_B, [\qq_B],
G-\vv)=\hat{\qq}_B-C(G({\hat{\qq}}_B)-\vv)+(I-CG'([\qq_B]))([\qq_B]-{\hat{\qq}}_B)\,\subset\,
\text{int}([\qq_B]).
\end{split}
\end{equation*}
According to Theorem \ref{thm:Ky}, for $\vv$ there exists a unique
root ${\qq}^*_B$ of (\ref{eq:square_int}) in $[\qq_B]$. Hence, for
each $\vv\in[\tilde{\vv}]$, there exists a unique root for the
system (\ref{eq:square_int}) if (\ref{Ky2}) holds.
\end{proof}


Alternatively, we also can apply another method provided
in \cite{chen2006existence}, to determine the existence of real
roots for the underdetermined system (\ref{int:underdetermined1})
directly. The only difference is that we need deal with a special
interval underdetermined system while they worked on an exact one.
For the same reason as in the above discussion, it is easy to generalize their
method in \cite{chen2006existence} to deal with  our problem.

Suppose the Jacobian $F'(\hat{\qq})$ is of full row rank. Following
\cite{chen2006existence}, we apply the column pivoting
$QR$-decomposition to choose an index set $B=\{k_1,k_2,...,k_s\}$
such that $F'_{B}({\hat{\qq}})\in\RR^{s\times s}$ is nonsingular.
Then, define the function $H:\RR^r\rightarrow\RR^r$ by
\begin{eqnarray*}\label{H_form2}
\left\{\begin{array}{l@{}l}
   &H_{B}(\qq)=\qq_{B}-F'_{B}({\hat{\qq}})^{-1}(F(\qq)-\vv),\\
   &H_{N}(\qq)=\qq_{N}-\alpha(\qq_{N}-{\hat{\qq}}_{N}),
\end{array}\right.
\end{eqnarray*}
where $N=\{1,2,...r\}/B$ and $\alpha\in(0,1)$ is a constant.
Obviously, if ${\qq}^*\in\RR^r$ is a fixed point of $H$, that is
$H({\qq}^*)={\qq}^*$, then we have $F({\qq}^*)-\vv=0$ with
${\qq}^*_N={\hat{\qq}}_N$. Choose two nonnegative numbers $r_1$ and
$r_2$, we define the convex set
$$
  [\qq]=\{\qq\in\RR^r: \|\qq_B-{\hat{\qq}}_B\|\leq r_1,\, \|\qq_N-{\hat{\qq}}_N\|\leq
  r_2\}.
$$
Now, we can use the following theorem to determine the existence of
solutions to the system (\ref{int:underdetermined1}).

\begin{theorem}\label{cor:solution}
Consider the system (\ref{int:underdetermined1}). Suppose the Jacobian
$F'(\hat{\qq})$ has full row rank, and that
$$
 \|F'_B(\qq)-F'_B({\hat{\qq}})\|\leq K\|\qq-{\hat{\qq}}\|  \text{ for } \qq\in [\qq].
$$
There is a solution ${\qq}^*$ of (\ref{int:underdetermined1}) in
$[\qq]$  for each $\vv\in[\vv]$ if
\begin{equation*}\label{cor:condition}
\begin{split}
  &\max_{\vv\in [\vv]}{\|F'_B({\hat{\qq}})^{-1}(F({\hat{\qq}})-\vv)\|}+\|F'_B({\hat{\qq}})^{-1}\| (\frac{1}{2}K(r_1+r_2)r_1+\max_{\qq\in [\qq]}\|F'_N(\qq)\|r_2)\leq r_1.
\end{split}
\end{equation*}
\end{theorem}
Remark that Theorem \ref{cor:square} is a special case of Theorem
\ref{cor:solution} by setting $r_2=0$. Compared with
Theorem~\ref{cor:solution}, the condition (\ref{Ky2}) in Theorem
\ref{cor:square} is easy to verify in practice.

\end{document}